\documentclass[11pt,english,american,fleqn]{article}

\usepackage{amsmath,amsfonts,amssymb,amsthm,mathrsfs}

\usepackage{bbm}

\usepackage{setspace}

\usepackage{enumerate}
\usepackage{bm}
\usepackage{bbm}
\usepackage{verbatim}
\usepackage{xargs} 
\usepackage{authblk}
\usepackage[square,sort,comma,numbers]{natbib}
\usepackage{chngcntr}
\usepackage[font=small]{caption}
\usepackage[hypertexnames=false,hidelinks]{hyperref}
\usepackage{dsfont}

\theoremstyle{plain}\newtheorem{theorem}{Theorem}[section]
\theoremstyle{plain}\newtheorem{lemma}[theorem]{Lemma}
\theoremstyle{plain}\newtheorem{corollary}[theorem]{Corollary}
\theoremstyle{plain}
\theoremstyle{plain}\newtheorem{proposition}[theorem]{Proposition}
\theoremstyle{definition}
\theoremstyle{remark}\newtheorem{remark}{Remark}
\theoremstyle{definition}\newtheorem{def:and:lemma}[theorem]{Definition and Lemma}
\theoremstyle{plain}

\numberwithin{equation}{section}
\counterwithin{remark}{section}
\counterwithin{figure}{section}

\newcounter{remarks}
\newcommand{\D}{\textnormal{d}}
\newcommand{\lsp}{\big \langle }
\newcommand{\rsp}{\big \rangle }

\newcommand{\re}{\operatorname{Re}}

\newcommand{\sno}{\vert\hspace{-1pt}\vert}

\newcommand{\eps}{\varepsilon}

\newcommand{\ud}{\mathrm{d}}

\newcommand{\ue}{\mathrm{e}}
\newcommand{\ui}{\mathrm{i}}
\newcommand{\R}{\mathbb{R}}
\newcommand{\Z}{\mathbb{Z}}

\newcommand{\N}{\mathbb{N}}
\newcommand{\cF}{\mathcal{F}}

\usepackage[colorinlistoftodos,prependcaption,textsize=tiny]{todonotes}

\usepackage{xcolor}

\begin{document}

\bibliographystyle{alpha}

\title{\Huge On the global minimum of the energy- momentum relation for the polaron}

\author{Jonas Lampart\thanks{CNRS \& LICB (UMR 6303 CNRS \& UBFC), 9. Av. Alain Savary, 21078 Dijon, France.\\ Email: \texttt{lampart@math.cnrs.fr}},\phantom{i} David Mitrouskas\thanks{Institute of Science and Technology Austria (ISTA), Am Campus 1, 3400 Klosterneuburg, Austria.\\ Email: \texttt{david.mitrouskas@ist.ac.at}}\phantom{i} and Krzysztof My\' sliwy\thanks{Institute of Science and Technology Austria (ISTA), Am Campus 1, 3400 Klosterneuburg, Austria.\\ Email: \texttt{krzysztof.mysliwy@ist.ac.at}}
}

\date{July 11, 2022}
\maketitle

\frenchspacing

\begin{spacing}{1.15} 
 
\begin{abstract}
For the Fr\"ohlich model of the large polaron, we prove that the ground state energy as a function
of the total momentum has a unique global minimum at momentum zero. This implies the non-existence of a ground state of the Fr\"ohlich Hamiltonian and thus excludes the possibility of a localization transition at finite coupling.
\end{abstract}

\allowdisplaybreaks

\section{Introduction}

The polaron describes an electron interacting with the quantized optical modes of a polarizable lattice. It poses a classic problem in solid state physics that was initiated by Landau's one-page paper  \cite{Landau1933} about the possibility of self-trapping of an electron by way of deformation of the lattice. In fact, one of our results shall be that for the large polaron, where the lattice is described by a continuous non-relativistic quantum field, this type of self-localization is absent for any finite value of the coupling.

Following H. Fr\"ohlich \cite{Froehlich1933,Froehlich1954} the large polaron is defined on the Hilbert space
\begin{align}
\mathscr{H} \, =\,  L^2(\mathbb R^3,\D x) \otimes \mathcal F
\end{align}
with $\mathcal F= \bigoplus_{n=0}^\infty \bigotimes^n_{\rm sym} L^2 (\mathbb R^{3})$ the bosonic Fock space, and governed by the Hamiltonian
\begin{align}\label{eq: Froehlich Hamiltonian}
H \,  =\,  -\Delta_x + N + \sqrt{\alpha} \phi(v_x),
\end{align}
where $-\Delta_x = (-\ui \nabla_x)^2$, $x$ and $-\ui\nabla_x$ describing the position and momentum of the electron and $N = d \Gamma(1)$ denotes the number operator on Fock space (modeling the energy of the phonon field, whose dispersion relation is constant). The interaction is described by the linear field operator
\begin{align}
\label{eq: def of h_x(y)}
\phi(v_x) =a(v_x)+a^*(v_x)
\quad \text{with} \quad v_x(y) = v(y-x), \quad v(y) =  \frac{1}{|y|^2}.
\end{align}
The bosonic creation and annihilation operators, $a^*$, $a$, satisfy the usual  canonical commutation relations.
After setting $\hbar = 1$ and the mass of the electron equal to $1/2$, the polaron model depends on a single dimensionless parameter $\alpha>0$. In the present work, the precise choice of $\alpha$ is not relevant, as our statements shall hold equally for all values of the coupling constant.

An important property of the Fr\"ohlich Hamiltonian $H$ is that it defines a translation-invariant system, in the sense that it commutes with the total momentum operator
\begin{align}
[H, -\ui \nabla_x + P_f] \, = \, 0
\end{align}
where $P_f = \D \Gamma(-\ui\nabla)$ describes the momentum of the phonons. As a consequence, it is possible to simultaneously diagonalize the total momentum and the energy. This simultaneous diagonalization is implemented best by the Lee--Low--Pines transformation \cite{LeeLowPines} $S : \mathscr H \to \int^{\oplus}_{\mathbb R^3}  \mathcal F\, dP$ defined by $F \circ e^{\ui P_fx}$ where $F$ indicates the Fourier transformation w.r.t. $x$. As is easily revealed by a direct computation, $S  (- \ui\nabla_x + P_f) S^* =  \int_{\mathbb R^3}^{\oplus} P\, dP$ and $S  H S^* =\int_{\mathbb R^3}^\oplus H(P)\, dP$ with
\begin{align}\label{eq:fiber:Hamiltonian}
 H(P) = (P - P_f)^2 + N + \sqrt \alpha \phi(v).
\end{align}
The fiber operator $H(P)$ acts on $\mathcal F$ and describes the system at total momentum $P\in \mathbb R^3$.

Irrespectively of whether we take \eqref{eq: Froehlich Hamiltonian} or \eqref{eq:fiber:Hamiltonian} as our starting point, the definition of the model is somewhat formal, since $v\notin L^2$. A common way to define the Fr\"ohlich model is to start from the quadratic form associated with $H(P)$ (or $H$) and show that it is an infinitesimal perturbation of the form associated with the non-interacting operator ($\alpha=0$). By the KLMN Theorem this implies the existence of a unique self-adjoint operator associated with the form of $H(P)$ (or $H$), which we call the fibered Fr\"ohlich  Hamiltonian. For convenience of the reader we give the details in the appendix.

The main topic of this paper is the energy-momentum relation of the polaron, defined as the lowest possible energy as a function of the total momentum
\begin{align}
E_\alpha(P) : = \inf \sigma(H(P)).
\end{align}
By rotation invariance of the Fr\"ohlich Hamiltonian, $E_\alpha(P) = E_\alpha(RP)$ for all rotations $R\in O(3)$ and hence $E_\alpha(P) = E_{{\rm rad},\alpha}(|P|) $. It was shown by L. Gross \cite{Gross1970} that $E_\alpha(0) \le E_\alpha(P)$ for all $P\in \mathbb R^3$. As recently pointed out in a paper by Dybalski and Spohn \cite{DybalskiS2020}, it is further conjectured that $P=0$ is the unique global minimum
\begin{align}\label{eq:strict inequality}
E_\alpha(0) < E_\alpha(P)\quad \text{for all} \quad P \neq 0.
\end{align}
While this inequality is not surprising from the physics point of view, its mathematical proof turns out to be less obvious. In fact already Gross' proof for showing that $P=0$ is a global minimum requires a considerable amount of work. In a small neighborhood around zero the strict inequality can be inferred from the finiteness of the effective mass \cite{DybalskiS2020} (the effective mass is defined by $(M_{\rm eff})^{-1} = \partial_{|P|}^2E_{{\rm rad},\alpha}(|P|)$). The goal of this paper is to prove \eqref{eq:strict inequality} for all non-zero $P\in \mathbb R^3$. 
A proof was proposed in 1991 by Gerlach and L\"owen~\cite{Gerlach91}. However, their argument is not completely rigorous and its validity has been recently debated~\cite{DybalskiS2020}.
Our aim is to follow the idea suggested by Gerlach and L\"owen and turn it into a complete mathematical proof. To this end we also benefit from results on the Fr\"ohlich Hamiltonian that became available more recently. In particular, we make use of the localization of the essential spectrum of the fiber Hamiltonians by the HVZ Theorem~\cite{Moeller2006}.

Approaching the problem from a different angle, Dybalski and Spohn provide a proof of \eqref{eq:strict inequality} under the assumption that a central limit theorem for the polaron path measure with two-sided pinning holds true, see \cite[Sec. 5]{DybalskiS2020}, in particular Theorem 5.3 and Conjecture 5.2 therein. To our knowledge, this particular version of the central limit theorem, however, has not been rigorously established so far.

At the time of finalizing our paper, we were informed by Polzer about his new work \cite{Polzer22} that presents a proof of \eqref{eq:strict inequality} based on a strategy, also in the probabilistic framework, that does not rely on the validity of the central limit theorem  with two-sided pinning. In this new approach the polaron path measure is represented as a mixture of Gaussian measures and analyzed with techniques from renewal theory (see also \cite{Betz2022,Betz2021,MV19}). Among other interesting properties, \cite{Polzer22} shows that $E_{{\rm rad},\alpha}(|P|)$ is a monotone increasing function with strict monotonicity below the essential spectrum, i.e.,\ in particular in a neighborhood around $P=0$.

In our proof we follow an operator-theoretic approach that is inspired by  \cite{Lampart2021,Gerlach91,Froehlich1974,Moeller2005,Faris72,LaSch19}. Except for the following lemma, whose proof is based on results from \cite{Moeller2006} (see Section \ref{sec:proof}), our presentation is fully self-contained.\footnote{In this regard, let us note that we shall not make use of Gross' inequality $E_\alpha(0)\le E_\alpha(P)$ \cite{Gross1970}, which would evidently imply property (ii).}

\begin{lemma}\label{lem:properties} Let $\alpha >0$ and set $E_\alpha(P) = \inf \sigma(H(P))$.
\begin{itemize}
\item [$\rm (i)$] $\sigma_{\rm ess}(H(P)) = [E_\alpha(0) + 1,\infty)$ for all $P\in \mathbb R^3$ 
\item[$\rm (ii)$] The set of global minima of $P \mapsto E_\alpha(P)$ is non-empty.
\end{itemize}
\end{lemma}


Before we continue with our main results, let us give some remarks on the context of the problem. For the non-interacting model, i.e. for $\alpha = 0$, one easily sees that $E_0(P) = \min \{ P^2 , 1 \}$ with $P^2$ describing the region where the full momentum is on the electron whereas in the constant region it is favorable that the total momentum is carried by a single phonon with the electron at rest. The expectation from physics is that the form of $E_\alpha(P)$ remains qualitatively the same also for non-zero $\alpha$, in particular that it remains essentially a parabola for a certain range of momenta around zero (though with a negative shift $E_\alpha(0)$ and with different curvature at $P=0$) and that it approaches $E_\alpha(0)+1$ as $|P|\to \infty$. That the expected parabola gives in fact an exact upper bound, $E_\alpha(P) \le E_\alpha(0) + P^2/(2M_{{\rm eff},\alpha})$ for all $P$ and $\alpha$, was shown in \cite{Polzer22}. The mathematical verification of the described picture is particularly interesting in the strong coupling limit $\alpha \to \infty$, where one can analyze the asymptotic expansion of $E_\alpha(P)$ and $M_{{\rm eff} , \alpha}$. For results in this direction see \cite{Lieb2019,Mit2022,Betz2022,Donsker1983,Lieb1997,Spohn1987}.

Another topic that is related to the energy-momentum relation and was debated for a long time in the physics literature is the existence of a ground state of $H$; see \cite{Peeters82,Gerlach88,Gerlach91} for a detailed discussion and further references. For weak coupling the non-existence is physically evident since the interaction term in $H$ is a small perturbation to the translation-invariant non-interacting part. For strong coupling, on the other hand, there is a leading-order contribution of the classical polarization field and hence the interaction cannot be viewed as a perturbation of a translation-invariant Hamiltonian anymore. In fact, in the limit of strong coupling, the ground state energy approaches the corresponding semiclassical energy \cite{Lieb1997,Donsker1983,Pekar54} which is known to possess a manifold of localized ground states \cite{Lieb1977}. This supported the idea that also for the full quantum model, the translational symmetry may be broken on the level of the ground state for sufficiently large coupling, thus resulting in a transition from a delocalized regime (weak coupling) to a localized one (strong coupling). The existence of a ground state for $H$, however, would require that the set of global minima of $E_\alpha(P)$ is a set of non-zero Lebesgue measure. As our result implies that the global minimum is unique, this possibility is ruled out and thus we obtain a rigorous proof (see Corollary \ref{thm:main:result}) of the non-existence of such a localization transition at any finite value of the coupling parameter.

\subsection{Main results}

The next statement is our main result.

\begin{theorem}\label{thm:global:minimum}
The energy-momentum relation $E_\alpha (P) = \inf \sigma (H(P))$ has its global minimum at $P=0$, i.e. for every $\alpha>0$ and every $P\neq 0$, we have the strict inequality
\begin{align}\label{eq:global:minimum}
E_\alpha (0) < E_\alpha(P).
\end{align}
\end{theorem}
\begin{remark} 
Note that our choice of the sign of $v$ in \eqref{eq: def of h_x(y)} does not matter, as the operator with form factor $\tilde v(y)=e^{i\theta} v(y)$ is unitarily equivalent to the one with $v$ via $\Gamma(e^{i\theta})$. Moreover, our proof can easily be generalized to all positive (up to a global phase) $v$ satisfying $\hat v(k) (k^2+1)^{-s}\in L^2(\mathbb R^3)$ for some $0\leq s<1/2$. For the sake of conciseness we focus on \eqref{eq: def of h_x(y)}
\end{remark}

As a corollary we obtain that the Fr\"ohlich Hamiltonian $H$ does not possess a ground state, irrespective of the value of the coupling constant. 
\begin{corollary}\label{thm:main:result}
There exists no $\alpha\in (0,\infty)$ such that $\inf \sigma(H) $ is an eigenvalue of $H$.
\end{corollary}
\begin{proof}[Proof]
Suppose on the contrary that there exists an $\alpha\in (0,\infty)$ such that $H$ has a ground state $\Psi_\alpha \in \mathscr H$. By the Lee--Low--Pines transformation (see introduction),
\begin{align}\label{eq:contra}
\inf \sigma(H) =  \frac{\int_{\mathbb R^3} dP \langle \Psi_{\alpha}(P) , H(P) \Psi_{\alpha}(P) \rangle }{\int_{\mathbb R^3} dP \|\Psi_{\alpha}(P) \|^2} \geq \frac{\int_{\mathbb R^3} dP \, E_\alpha (P) \|\Psi_{\alpha}(P)\|^2}{\int_{\mathbb R^3} dP \|\Psi_{\alpha}(P)\|^2}  
\end{align}
with $\Psi_{\alpha}(P) = F[e^{\ui P_fx} \Psi_\alpha](P) \in \mathcal F$ where $F$ is the Fourier transform w.r.t $x$. Note that $\|\Psi_{\alpha}(P) \|^2 \in L^1(\mathbb R^3)$ by Parseval and $e^{\ui P_fx} \Psi_\alpha \in L^2(\mathbb R^3,\mathcal F)$. Now by Theorem \ref{thm:global:minimum} the set $\{ P \in \mathbb R^3 \, | \, E_\alpha(P) =E_\alpha(0)\}$ has Lebesgue measure zero, and hence the right side of \eqref{eq:contra} is strictly larger than $E_\alpha(0) $. On the other hand, \eqref{eq:global:minimum} together with the fiber decomposition of $H$ imply that $ \inf \sigma(H) = E_\alpha(0) $, which yields a contradiction.
\end{proof}

\section{Proof of Theorem \ref{thm:global:minimum}}\label{sec:proof}

To prove Theorem \ref{thm:global:minimum} we introduce a suitable auxiliary Hamiltonian and analyze the degeneracy of its ground state. On the one hand, it shall be easily seen that the multiplicity is strictly larger than one if we assume that $P\mapsto E_\alpha(P)$ attains its global minimum for some $q\neq 0$. On the other hand, this cannot be the case as it will follow by an abstract  result that the ground state of this auxiliary operator is non-degenerate. As mentioned in the introduction, this idea goes back to Gerlach and L\"owen \cite{Gerlach91}.

By Lemma \ref{lem:properties} (ii) there exists a $q \in \mathbb R^3$ such that $E_\alpha(q)\le E_\alpha(P)$ for all $P\in \mathbb R^3$. For said $q$ we choose $\ell =\ell(q) >0$ such that $q\in (2\pi/\ell)\mathbb Z^3$ and consider the Hilbert space
\begin{align}
\mathscr H_\ell = L^2(\Lambda_\ell) \otimes \mathcal F,
\end{align}
where $\Lambda_\ell = [0,\ell]^3$ is the three-dimensional cube of side length $\ell$ with periodic boundary conditions, and $\mathcal F$ denotes again the bosonic Fock space over $L^2(\mathbb R^3)$. On this space we define the Hamiltonian
\begin{align}\label{eq: auxiliary Hamiltonian}
h_\ell = (-\ui \nabla_x  -  P_f)^2 +  N + \sqrt{\alpha} \phi(v) 
\end{align}
where $-\ui \nabla_x$ now denotes the gradient on $L^2(\Lambda_\ell)$. 
The reason for restricting the variable $x$ to the torus is that $h_\ell$ now has a ground state. 
Indeed, since $x$ only appears in the gradient, and since the corresponding momenta are restricted to $(2\pi / \ell ) \mathbb Z^3$, one can readily diagonalize
\begin{align}\label{eq:direct:sum:representation:hell}
 h_\ell= \bigoplus_{P\in \frac{2\pi}{\ell} \mathbb Z^3} \pi_{\ell,P} \otimes H(P)
\end{align}
with $\pi_{\ell,P} = |\varphi_{\ell,P}\rangle \langle \varphi_{\ell,P}|$ the projection onto the normalized plane wave $\varphi_{\ell,P}=\ell^{-1/2} e^{-\ui Px} \in L^2(\Lambda_\ell)$ with momentum $P$. Here it is important to note that the Fock space component is described by the fiber Hamiltonian $H(P)$ as defined in \eqref{eq:fiber:Hamiltonian}. Also note that since $H(P)$, $P\in \mathbb R^3$, is self-adjoint and bounded from below, so is $h_\ell$. By Lemma \ref{lem:properties} we know that $E_\alpha(q) < \inf \sigma_{\rm ess}(H(q))$, hence there is an eigenfunction $\xi(q) \in \mathcal F$ such that $H(q) \xi (q) = E_\alpha(q) \xi (q)$. Since $E_\alpha(q) \le E_\alpha(P)$ for all $P\in \mathbb R^3$, the wave function $\varphi_{\ell,q}  \otimes \xi(q) \in \mathscr H_{\ell}$ is a ground state of $h_\ell$ with eigenvalue $ E_\alpha(q)$. 

The next lemma, whose proof is postponed to the next section, shows that this ground state is actually unique.

\begin{proposition}\label{lem:ground_state:hL} For all $\alpha >0$ and any $\ell>0$ for which $\inf \sigma(h_\ell)$ is an eigenvalue (in particular for the value $\ell(q)$ chosen above), this eigenvalue has multiplicity one.
\end{proposition}


With this at hand, we can prove Theorem \ref{thm:global:minimum} by way of contradiction. To this end assume that $q\neq 0$. By rotational symmetry $E_\alpha(q) = E_\alpha(-q)$ and, since $E_\alpha(-q)<\inf\sigma_{\rm ess}(H(-q))$ by Lemma \ref{lem:properties}, the fiber Hamiltonian $H(-q)$ has an eigenfunction $\xi(-q) \in \mathcal F$ with eigenvalue $E_\alpha(-q)$. This implies that there is a second ground state of $h_\ell$, given by $\varphi_{\ell,-q} \otimes \xi(-q)\in \mathscr H_\ell$, that is orthogonal to $\varphi_{\ell,q} \otimes \xi(q)$. Thus we arrive at a contradiction to the statement of Proposition \ref{lem:ground_state:hL}, and therefore we can rule out that $q$ is non-zero. This shows that $P\mapsto E_\alpha(P)$ can attain its global minimum only at $P=0$, hence the proof of Theorem \ref{thm:global:minimum} is complete.\medskip

In the remainder of this subsection we prove Lemma \ref{lem:properties}. For that purpose we first state and prove the following general result.

\begin{lemma}\label{lem:aux:nrc} Let $A_n$, $n\in \mathbb N$, be a sequence of self-adjoint operators satisfying  $\sigma_{\rm ess}(A_n) = [a_n,b_n]$ with $- \infty \le a_n \le b_n \le \infty$. If $A_n\xrightarrow{n\to \infty}A$ in norm-resolvent sense for some self-adjoint operator $A$, then $a_n\xrightarrow{n\to \infty} a$, $b_n\xrightarrow{n\to \infty} b$ for some $- \infty \le a \le b \le \infty$ and $\sigma_{\rm ess}(A) = [a,b]$.
\end{lemma}
\begin{proof}[Proof of Lemma \ref{lem:aux:nrc}] First assume that $a_n, b_n \xrightarrow{n\to \infty} a,b $ for some $-\infty \le a \le b \le \infty$. Since $\sigma_{\rm ess}(A_n)\neq \varnothing$, the essential spectrum of $A$ cannot be empty \cite[Thm. XIII.23]{RS1}. Now, for $\varepsilon>0$ we choose $N(\varepsilon)$ such that $[a_n,b_n] \subset [a-\varepsilon, b+\varepsilon]$ for all $n\ge N(\varepsilon)$. It follows that $\sigma(A_n)$ is purely discrete in $(-\infty,a-\varepsilon) \cup (b+\varepsilon,\infty)$ for all $n\ge N(\varepsilon)$. By \cite[Prop. 11.4.31]{Oliviera} this implies that $\sigma(A)$ is purely discrete on the same interval and thus, using that $\varepsilon$ was arbitrary, $\sigma_{\rm ess}(A)\subset [a,b]$. Since $\sigma_{\rm ess}(A)\neq \varnothing $, we get $\sigma_{\rm ess}(A)= \{a\}$ if $a=b$.

For $a<b$, let $t\in (a,b)\cap \rho(A)$. By norm-resolvent convergence there is a sequence $ t_n\in  \rho(A_n)$ such that $t_n\xrightarrow{n\to \infty} t$ \cite[Thm. XIII.23]{RS1}. Since $\rho(A_n) \subset [a_n,b_n]^c \xrightarrow{n\to \infty}  [a,b]^c$, this contradicts the assumption $t\in (a,b)$. Hence $(a,b)\subset \sigma(A)$, and since the spectrum is closed, $[a,b]\in \sigma(A)$. This proves $\sigma_{\rm ess}(A)=[a,b]$ also for $a<b$.

To remove the initial assumption on $a_n,b_n$, assume there is no $a\in[-\infty,\infty]$ such that $a_n \xrightarrow{n\to \infty}a$ (the same argument applies to $b_n$). Then there are at least two subsequences of $a_n$ with limits $a,\tilde a \in[-\infty,\infty]$, $a\neq \tilde a$. Applying the argument from above to the two subsequences (and the associated subsequences of $A_n$) separately leads to the contradiction that $[a,b] = \sigma_{\rm ess}(A) = [\tilde a,b]$. 
\end{proof}

\begin{proof}[Proof of Lemma \ref{lem:properties}] The proof is essentially based on the results from \cite{Moeller2006}. Since the latter considers more regular polaron models with $v\in L^2(\mathbb R^3)$, we need to introduce the fiber Hamiltonians with UV cutoff. We define $H_{\Lambda}(P)$ as in \eqref{eq:fiber:Hamiltonian} with form factor $v_\Lambda$ defined by $\hat v_\Lambda =  \chi_\Lambda \hat v$ where $\chi_\Lambda$ is the characteristic function $k\mapsto \chi_{[0,\Lambda)}(|k|)$. Since $v_\Lambda \in L^2(\mathbb R^3)$, it is easy to verify that $H_{\Lambda}(P)$ is self-adjoint and bounded from below. Denoting $E_{\alpha,\Lambda}(P)= \inf\sigma(H_\Lambda(P))$, we can quote two important results from \cite{Moeller2006}: For every $\Lambda>0$
\begin{itemize}
\item[(i')] $\sigma_{\rm ess}(H_{\alpha,\Lambda}(P)) = [E_{\alpha,\Lambda}(0)+1,\infty)$,
\item[(ii')] $\lim_{|P|\to \infty} | E_{\alpha,\Lambda}(P) - \inf \sigma_{\rm ess}(H_\Lambda(P)) | = 0$.
\end{itemize}
These statements are given in \cite[Thm. 2.1. and Thm. 2.4]{Moeller2006}. By Proposition \ref{prop:H-conv}, we also have $H_{\Lambda}(P) \xrightarrow{\Lambda \to \infty} H(P)$ in norm-resolvent sense, which implies $E_{\alpha,\Lambda}(0) \xrightarrow{\Lambda\to \infty} E_\alpha(0)$ \cite[Thm. XIII.23]{RS1}. Hence (i) is a consequence of (i') and Lemma \ref{lem:aux:nrc}.
 
By the Lee--Low--Pines transformation $\mathcal E := \inf_{P\in \mathbb R^3} E_\alpha(P)$ coincides with $\inf\sigma(H)$ and thus $\mathcal E>-\infty$. Item (ii) states that $\mathcal M = \{ P\in \mathbb R^3 \, |\, E_\alpha(P) = \mathcal E \}$ is a non-empty set. Assuming the opposite, $\mathcal M = \varnothing$, implies the existence of a sequence $(P_j)_{j\in \mathbb N} \subset \mathbb R^3$, $|P_j| \to \infty$, satisfying $\lim_{j\to \infty}E_\alpha(P_j) = \mathcal E$. W.l.o.g. we can choose $P_j$ such that $|P_j|$ is monotone increasing and $E_\alpha(P_j)\in (\mathcal E, \mathcal E + \tfrac{1}{4})$ for all $j\ge 1$. Further consider $ f \in C_c^\infty(\mathbb R)$ with $\text{supp}(f)\subseteq (\mathcal E -\tfrac{1}{2} , \mathcal E  + \tfrac{1}{2})$ and $f(s) = 1$ for all $s \in (\mathcal E -\tfrac{1}{4}, \mathcal E+\tfrac{1}{4} )$. By uniform norm-resolvent convergence of $H_\Lambda(P)$ (see Proposition \ref{prop:H-conv}) and \cite[Thm. VIII.20]{RS1}, we have
\begin{align}
0 & = \lim_{\Lambda \to \infty} \sup_{P\in \mathbb R^3} \sno f(H_\Lambda(P)) - f(H(P)) \sno \notag \\ 
& \ge  \lim_{\Lambda \to \infty} \sup_{j \ge j_\Lambda }  \sno f(H_\Lambda(P_j)) - f(H(P_j)) \sno = \lim_{\Lambda \to \infty} \sup_{j \ge j_\Lambda } \sno f(H(P_j)) \sno \label{eq:compactness:proof}
\end{align}
where the last step follows from (i') and (ii') if we choose $j_\Lambda$ large enough. In more detail, for every $\Lambda$, we can choose $j_\Lambda$ such that $\sigma(H_{\Lambda}(P_j)) \cap (-\infty,E_{\alpha,\Lambda}(0) + \tfrac{3}{4}) =\varnothing$ for all $j\ge j_\Lambda$. Since $\mathcal E \le E_\alpha(0) \le E_{\alpha,\Lambda}(0) + \tfrac{1}{4}$ (here we use again $E_{\alpha,\Lambda}(0) \xrightarrow{\Lambda \to \infty} E_\alpha(0)$), this implies $\sigma(H_{\Lambda}(P_j)) \cap (-\infty,\mathcal E+ \tfrac{1}{2}] = \varnothing$ for all $j\ge j_\Lambda$, and thus $f(H_\Lambda(P_j)) = 0$. By assumption, however, $\inf \sigma(H(P_j))\in (\mathcal E , \mathcal E + \tfrac{1}{4})$ for all $j\ge 1$ and thus the right side of \eqref{eq:compactness:proof} is one, which is a contradiction.
\end{proof}

\subsection{Proof of Proposition~\ref{lem:ground_state:hL}}

In this section we show that the resolvent of $h_\ell$ is a positivity improving operator w.r.t. a suitable Hilbert cone using a strategy that was applied to the renormalized Nelson model in~\cite{Lampart2021}. (Compared to the latter we face the additional difficulty that $h_\ell$ is an operator on the tensor product $\mathscr H_\ell$ and not only on $\mathcal F$). The statement of Proposition \ref{lem:ground_state:hL} will then follow from a Perron-Frobenius type argument due to Faris \cite{Faris72}. 
Let 
\begin{align}
\mathcal C := \big\{ \Psi \in \mathscr H_\ell \, \big\vert \, \forall n\in \N_0: (-1)^n \Psi^{(n)}(x,y_1,\ldots,y_n)\ge 0 \big\}.
\end{align}
This defines a (Hilbert) cone in $\mathscr H_\ell$ in the sense of~\cite{Faris72}.
A bounded operator $K$ on $\mathscr H_\ell$ is called
\begin{itemize}
 \item positivity preserving (with respect to $\mathcal C$) if $K u \in \mathcal C$ for any $u\in \mathcal C$, and
\item positivity improving (with respect to $\mathcal C$) if $\langle Ku, v \rangle>0$ for any $u,v\in \mathcal C  \setminus\{ 0 \}$.
\end{itemize}

Now consider the self-adjoint positive operator
\begin{equation}
\mathfrak h_{\ell,0} =  (-\ui \nabla_x   + \ud \Gamma(-\ui \nabla))^2 + N +1
\end{equation}
acting on the Hilbert space $\mathscr H_\ell = L^2(\Lambda_\ell) \otimes \mathcal F$.

\begin{lemma}\label{lem:h_0-pos}
The resolvent $\mathfrak{h}_{\ell,0}^{-1}$ is positivity preserving with respect to the cone $\mathcal C$. Its restriction to the zero-phonon space $L^2(\Lambda_\ell)$ is positivity improving with respect to the cone of positive functions in $L^2(\Lambda_\ell)$. 
\end{lemma}
\begin{proof}
 The restriction of $\mathfrak{h}_{\ell,0}^{-1}$ to $L^2(\Lambda_\ell)$ is given by the periodization of the resolvent of $-\Delta$ on $\R^3$, i.e. it has the integral kernel
 \begin{equation}
  \sum_{k\in \Z^3} \frac{\ue^{-|x-x'-k\ell|}}{4\pi |x-x'-k\ell|},
 \end{equation}
which converges due to the exponential decay of the numerator and is strictly positive.
Similarly, the integral kernel of $\mathfrak{h}_{\ell,0}^{-1}$ on the $n$-phonon space $\cF^{(n)}$ is (as one readily checks using the Fourier transform)
\begin{equation}
 \sum_{k\in \Z^3} \frac{\ue^{-\sqrt{n+1}|x-x'-k\ell|}}{4\pi |x-x'-k\ell|}\prod_{j=1}^n\delta(y_j-x-y'_j+x'+k\ell),
\end{equation}
which is a non-negative distribution.
\end{proof}

We recall that $v(y) = |y|^{-2}>0$ and $\hat v(k)=(4\pi |k|)^{-1}$.

\begin{lemma}\label{lem:ah^{-1}}
The operator $a(v)\mathfrak{h}_{\ell,0}^{-1}$ is bounded, and $1+a(v)\mathfrak{h}_{\ell,0}^{-1}$ is invertible, with bounded inverse given by
\begin{equation*}
( 1+a(v)\mathfrak{h}_{\ell,0}^{-1})^{-1}= \sum_{j=0}^\infty \Big(a(-v) \mathfrak{h}_{\ell,0}^{-1}\Big)^j,
\end{equation*}
where the sum converges in the operator norm.
\end{lemma}
\begin{proof}
For $n\in \N$ we have using the Fourier transform and Cauchy--Schwarz
\begin{align}
 &\|a(v)\mathfrak{h}_{\ell,0}^{-1/2} \Psi^{(n)}\|_{L^2(\Lambda_\ell)\otimes \cF^{(n-1)}}^2 \notag\\
 & = n\sum_{p\in \frac{2\pi}{\ell} \Z^3} \int_{\R^{3(n-1)}} dk_1\cdots dk_{n-1} \bigg| \int dk_n \frac{ \hat v(k_n) \Psi^{(n)}(p,k_1, \dots, k_n)}{((p+k_1+\cdots k_n)^2 + n+1)^{1/2}} \bigg|^2 \notag\\
 &\leq n \|\Psi^{(n)}\|^2 \sup_{k\in \R^3} \int d\xi\frac{|\hat v(\xi)|^2}{(k+\xi)^2 + n+1} \notag \\
 &\leq (n+1)^{1/2} \|\Psi^{(n)}\|^2 \frac{1}{(4\pi)^2}\int \frac{d\eta}{\eta^2(\eta^2+1)}.\label{eq:a-integral}
\end{align}
We thus have the inequality 
\begin{equation}\label{eq:a-bound}
 \| a(v)\mathfrak{h}_{\ell,0}^{-1/2} \Psi\| \leq C \| (N+1)^{1/4}\Psi\|,
 \end{equation}
 which implies that
 \begin{equation}
 \| a(v)\mathfrak{h}_{\ell,0}^{-1} \Psi\| \leq C \| (N+1)^{-1/4}\Psi\|. 
 \end{equation}
 For $j\in \N$  this implies
\begin{equation}
 \| (a(v)\mathfrak{h}_{\ell,0}^{-1})^j \Psi \| \leq C \| (a(v)\mathfrak{h}_{\ell,0}^{-1})^{j-1}(N+j)^{-1/4} \Psi\|
 \leq C^j (j!)^{-1/4} \|\Psi\|.
\end{equation}
Since $C^j (j!)^{-1/4}$ is summable, the Neumann series converges and one easily checks that the limit is the inverse of $1+a(v)\mathfrak{h}_{\ell,0}^{-1}$.
\end{proof}

With the operator $\mathfrak{h}_{\ell,0}$ we can write the auxiliary operator \eqref{eq: auxiliary Hamiltonian} as
\begin{equation*}
h_\ell + 1 = \underbrace{(1+a(v)\mathfrak{h}_{\ell,0}^{-1})\mathfrak{h}_{\ell,0} (1+\mathfrak{h}_{\ell,0}^{-1}a^*(v))}_{=:K} \underbrace{- a(v) \mathfrak{h}_{\ell,0}^{-1} a^*(v)}_{=:T} = K+T.
\end{equation*}

\begin{proposition}\label{prop:positivity}
Let $\alpha>0$. For all $\lambda>-\inf\sigma(h_\ell)$ the resolvent $(h_\ell+\lambda)^{-1}$ is positivity improving with respect to $\mathcal{C}$. 
\end{proposition}

\begin{proof}
It is sufficient to prove the statement for one $\lambda > -\inf\sigma(h_\ell)$, by analyticity of the resolvent (see also \cite[Lem. A.4]{Moeller2006}). For sufficiently large $\lambda$ we shall first prove
 \begin{align}\label{eq:Neumann-Reihe}
  & (h_\ell + 1 + \lambda)^{-1} = (K+T+\lambda)^{-1} \notag\\
  &  \quad\quad = (K +  \lambda)^{-1} \Big(1 + T(K + \lambda)^{-1}\Big)^{-1} = (K + \lambda)^{-1} \sum_{j=0}^\infty  \Big(- T(K + \lambda)^{-1}\Big)^j.
 \end{align}
To this end, we show that $T$ is infinitesimally $K$-bounded. In fact, by~\eqref{eq:a-bound},
\begin{equation}\label{eq:T-bound}
 \|T \Psi \| \leq C\| (N+1)^{1/2}\Psi\| \leq \eps \|N \Psi\| + C_\eps \|\Psi\|,
\end{equation}
for arbitrary $\eps>0$, and moreover,
\begin{align}
 \|N \Psi \| &\leq \|N (1+\mathfrak{h}_{\ell,0}^{-1}a^*(v)) \Psi \| + \|N \mathfrak{h}_{\ell,0}^{-1}a^*(v) \Psi \| \notag\\
 &\leq \|\mathfrak{h}_{\ell,0} (1+\mathfrak{h}_{\ell,0}^{-1}a^*(v)) \Psi \| +  C \|(N+1)^{3/4} \Psi \| \notag \\
 & \leq \|(1+a(v)\mathfrak{h}_{\ell,0}^{-1})^{-1}\| \|K\Psi\| + \tfrac34 \|N \Psi \| + (\tfrac14 C^4  +\tfrac34)\|\Psi\|. \label{eq:N-bound}
\end{align}

Proceeding with \eqref{eq:Neumann-Reihe}, we use that $-T$ is positivity preserving, since it is the product of the three positivity preserving operators $\mathfrak{h}_{\ell,0}^{-1}$ (see Lemma~\ref{lem:h_0-pos}) and $a(-v)$, $a^*(-v)$ (note the alternating sign in the definition of $\mathcal{C}$). The claim will thus follow if we can prove that $(K + \lambda )^{-1}$ improves positivity. As $K\geq 1$, this holds for all $\lambda>-1$ if it holds for $\lambda=0$.

Let $0\neq \Psi,\Phi \in \mathcal{C}$ and let $n,m$ be so that $\Psi^{(n)}\neq 0$, $\Phi^{(m)}\neq 0$. 
Set $\phi:= (a(-v)\mathfrak{h}_{\ell,0}^{-1})^m \Phi^{(m)}$ and $\psi:=(a(-v)\mathfrak{h}_{\ell,0}^{-1})^n \Psi^{(n)}$. Since $a(-v)$ and $\mathfrak{h}_{\ell,0}^{-1}$ are positivity preserving, $\phi$, $\psi$ are  non-negative elements of $L^2(\Lambda_\ell)$. Since  $v>0$, they do not vanish identically. 
We thus have
\begin{align}
 \langle \Phi, K^{-1} \Psi \rangle &= \bigg\langle \sum_{j=0}^\infty \big(a(-v)( \mathfrak{h}_{\ell,0})^{-1}\big)^j \Phi, \mathfrak{h}_{\ell,0}^{-1} \sum_{k=0}^\infty \big(a(-v)\mathfrak{h}_{\ell,0}^{-1}\big)^k \Psi \bigg\rangle 
 \ge  \langle \phi,\mathfrak{h}_{\ell,0}^{-1}  \psi \rangle >0,
 \end{align}
since $\mathfrak{h}_{\ell,0}$ improves positivity on $L^2(\Lambda_\ell)$ by Lemma~\ref{lem:h_0-pos}. This proves the claim.
\end{proof}

\begin{proof}[Proof of Proposition~\ref{lem:ground_state:hL}]
 The proposition is proved using the positivity-improving property of $h_\ell$ following arguments of \cite{Faris72}. Let $\mathscr{H}_{\ell}^\R$
be the subspace of real-valued functions in $\mathscr{H}_\ell$. Since $h_\ell$ is invariant under complex conjugation, it can be restricted to an operator on $\mathscr{H}_\ell^\mathbb R $.
Moreover, any eigenvalue $e$ of $h_\ell$ is also an eigenvalue of this restriction and their multiplicity is the same, since the $\R$-linear map of multiplication by the imaginary unit is an isomorphism between the real and imaginary subspaces of $\ker(h_\ell-e)$.
 
It is thus sufficient to prove that $\inf \sigma(h_\ell)$ is a simple eigenvalue of $h_\ell\vert_{\mathscr{H}_\ell^\R}$. This is equivalent to proving that, for $\lambda>-\inf\sigma(h_\ell)$, the space of real eigenfunctions of $(h_\ell+\lambda)^{-1}$ for its largest eigenvalue $e:=(\inf\sigma(h_\ell)+\lambda)^{-1}>0$ has dimension one.

Let $\Psi\in \mathscr{H}_\ell^\R$ be a normalized eigenfunction of $(h_\ell+\lambda)^{-1}$ with eigenvalue $e$. We can write
\begin{equation}
 \Psi=\Psi_+ - \Psi_-,
\end{equation}
where 
\begin{align}
 \Psi_\pm^{(n)}=(-1)^{n} \max\{ \pm \Psi^{(n)},0\}, 
\end{align}
so $\Psi_\pm\in \mathcal{C}$ and $\langle \Psi_+, \Psi_-\rangle=0$. We then have
\begin{align}
 e&=\langle \Psi, (h_\ell+\lambda)^{-1} \Psi \rangle \notag \\
 &= \langle \Psi_+, (h_\ell+\lambda)^{-1} \Psi_+ \rangle + \langle \Psi_-, (h_\ell+\lambda)^{-1} \Psi_- \rangle - 2\underbrace{ \langle \Psi_-, (h_\ell+\lambda)^{-1} \Psi_+ \rangle }_{\geq 0}\notag \\
 &\leq \langle (\Psi_++\Psi_-), (h_\ell+\lambda)^{-1} (\Psi_++\Psi_-) \rangle \leq e, \label{eq:positivity-ev}
\end{align}
since $e$ is the largest eigenvalue and $\|\Psi_++\Psi_-\|=\|\Psi\|=1$. We must thus have equality in~\eqref{eq:positivity-ev}, so
\begin{equation}
\langle \Psi_-, (h_\ell+\lambda)^{-1} \Psi_+ \rangle = 0.
\end{equation}
Since $(h_\ell+\lambda)^{-1}$ improves positivity  this implies that either $\Psi_+$ or $\Psi_-$ are equal to zero, i.e. $\Psi \in \mathcal{C}$ or $-\Psi\in \mathcal{C}$. 

Now assume there exist two orthogonal real eigenfunctions $\Phi, \Psi \in \ker((h_\ell+\lambda)^{-1}-e)$. By changing signs if necessary, we may assume that $\Phi,\Psi\in \mathcal{C}\setminus\{0\}$. Then
\begin{equation}
 \langle \Phi, \Psi \rangle = e^{-1}  \langle \Phi, (h_\ell+\lambda)^{-1}\Psi \rangle >0,
\end{equation}
a contradiction, so $e$ is a simple eigenvalue.
\end{proof}

\appendix

\section{Definition of the Fröhlich Hamiltonian}

In this section we define the Fröhlich Hamiltonian as a self-adjoint operator following the ideas of~\cite{LaSch19}. For a different proof, based on a commutator trick introduced in~\cite{Lieb1997,Lieb1958}, see~\cite{griesemerwuensch}. 

We will give the construction for the fiber operators $H(P)$, from which $H$ can be obtained by reversing the Lee--Low--Pines transformation. Consider the quadratic form
\begin{equation}
 Q_P(\Psi):=\langle \Psi,(H(P)+1)\Psi\rangle  = \lsp \Psi, ((P-P_f)^2 + N  +1) \Psi \rsp + 2 \sqrt \alpha \re \lsp \Psi , a(v) \Psi \rsp.
\end{equation}
Introducing the positive operator $H_0:=(P-P_f)^2 + N  +1$, one easily sees that $Q_P$ is well defined on $D(H_0^{1/2})$. Moreover, by inequality~\eqref{eq:aH-bound} below, $Q_P$ is an infinitesimal perturbation of the form of $H_0$.
Consequently, by the KLMN Theorem there exists a unique self-adjoint operator whose form is $Q_P$.
To make this more explicit, rewrite $Q_P$ as
\begin{equation}
 Q_P(\Psi) = \langle \Psi,(1+a(v)H_0^{-1})H_0(1+H_0^{-1}a^*(v))\Psi\rangle  - \langle \Psi, a(v)H_0^{-1} a^*(v)\Psi \rangle.
\end{equation}
From this, we see that $Q_P$ is associated with the operator
\begin{equation}\label{eq:H IBC}
 H(P)+1=\underbrace{(1+a(v)H_0^{-1})H_0(1+H_0^{-1}a^*(v))}_{=:K(P)} \underbrace{- a(v)H_0^{-1} a^*(v)}_{=:T(P)}=K(P)+T(P).
\end{equation}

\begin{proposition}\label{prop:sa}
 For $\alpha>0$ and $P\in \R^3$ the quadratic form $Q_P$ is the form of the self-adjoint operator $H(P)+1$ given by~\eqref{eq:H IBC} with domain
 \begin{equation*}
  D(H(P))=\big\{\Psi\in \cF \, \big\vert\, (1+H_0^{-1}a^*(v))\Psi \in D(H_0)\big\}.
 \end{equation*}
\end{proposition}

\begin{proof}
 By the same proof as in~\eqref{eq:a-bound} we have for $\Psi\in D(N^{1/4})$
 \begin{equation}\label{eq:aH-bound}
  \|a(v)H_0^{-1/2}\Psi\| \leq C \| (N+1)^{1/4}\Psi\|.
 \end{equation}
Hence, by the proof of Lemma~\ref{lem:ah^{-1}}, the operator $1+a(v)H_0^{-1}$ and its adjoint are boundedly invertible.
This implies that $D(H(P))$ is dense, since $D(H_0)$ is.
Moreover, $K(P)$ is a symmetric, invertible operator and thus self-adjoint.  
The proof is completed by showing that $T(P)$ is infinitesimally $K(P)$-bounded as in~\eqref{eq:T-bound},~\eqref{eq:N-bound}.
\end{proof}

In the proof of Lemma~\ref{lem:properties} we used that the operators with cutoff $H_\Lambda(P)$ converge to $H(P)$ in norm-resolvent sense.
We give a proof of this fact in the following proposition (see also~\cite{griesemerwuensch}).

\begin{proposition}\label{prop:H-conv}
 Let $\hat v_\Lambda(k) = \hat v(k) \chi_{[0,\Lambda)}(|k|)$, then the family of self-adjoint operators
 \begin{equation*}
  H_\Lambda(P)=(P-P_f)^2 + N + \phi(v_\Lambda)
 \end{equation*}
with domain $D(H_\Lambda(P))=D(H_0)$ converges to $H(P)$ in norm-resolvent sense as $\Lambda\to \infty$ uniformly in $P\in \R^3$.
\end{proposition}
\begin{proof}
Recall that $H_0 = (P-P_f)^2 + N + 1$. The bounds on the difference of the resolvents will manifestly be independent of $P$, so we will not emphasize this at every step. Let 
 \begin{equation}
  K_\Lambda(P):=(1+a(v_\Lambda)H_0^{-1})H_0(1+H_0^{-1}a^*(v_\Lambda))
 \end{equation}
and
\begin{equation}
 T_\Lambda(P) = - a(v_\Lambda)H_0^{-1}a^*(v_\Lambda),
\end{equation}
so that $ H_\Lambda(P)+1= K_\Lambda(P) + T_\Lambda(P)$, Similarly to in~\eqref{eq:H IBC}.
Then with $z =1\pm i$
\begin{align}
 &(H_\Lambda(P)+z)^{-1}- (H(P)+z)^{-1} = (K_\Lambda(P)+T_\Lambda(P)\pm i)^{-1}- (K(P)+T(P)\pm i)^{-1} \notag\\
 &= (H_\Lambda(P)+z)^{-1}(a(v-v_\Lambda)H_0^{-1} H_0(1+H_0^{-1}a^*(v)) (H(P)+z)^{-1} \label{eq:HK-difference} \\
 &\quad + (H_\Lambda(P)+z)^{-1}(1+a(v_\Lambda)H_0^{-1}) H_0 (1+H_0^{-1}a^*(v-v_\Lambda))) (H(P)+z)^{-1} \label{eq:HKLambda-difference} \\
 &\quad + (H_\Lambda(P)+z)^{-1}(T(P)-T_\Lambda(P)) (H(P)+z)^{-1}\label{eq:HT-difference}.
\end{align}
By an analogous bound to~\eqref{eq:a-integral}, we have
\begin{equation}\label{eq:a-difference}
 \|a(v-v_\Lambda)H_0^{-1/2}\Psi\| \leq \frac{1}{4\pi}\bigg(\int_{|k|>\Lambda} \frac{dk}{k^2(k^2+1)}\bigg)^{1/2} \|(N+1)^{1/2}\Psi\|,
\end{equation}
whence $\|a(v-v_\Lambda)H_0^{-1}\|$ tends to zero as $\Lambda\to \infty$.
Since $ H_0(1+H_0^{-1}a^*(v)) (H(P)+z)^{-1}$ is a bounded operator by construction of $H(P)$, this shows that~\eqref{eq:HK-difference} tends to zero in norm.

Since 
$|\hat v_\Lambda(k)|\leq |\hat v(k)|$ we have
\begin{equation}\label{eq:aLambda-bound}
 \|a(v_\Lambda)H_0^{-1/2}\Psi\|\leq C \|(N+1)^{1/4}\Psi\|  
\end{equation}
uniformly in $\Lambda$. With this, the bounds from the proof of Proposition~\ref{prop:sa} hold for $H_\Lambda(P)$ uniformly in $\Lambda$ and~\eqref{eq:HKLambda-difference} tends to zero for the same reason as~\eqref{eq:HK-difference}.

To conclude, note that by~\eqref{eq:a-difference} and~\eqref{eq:aLambda-bound} we have
\begin{equation}
 \|(T_\Lambda(P)-T(P)) \Psi\|\leq C_\Lambda \|(N+1)^{3/4}\Psi\|
\end{equation}
with $\lim_{\Lambda\to \infty}C_\Lambda=0$. Then, since $N$ is $H(P)$-bounded by construction of $H(P)$, the expression~\eqref{eq:HT-difference} also tends to zero in norm as $\Lambda\to \infty$, which proves the claim.
\end{proof}
%

\bigskip\vspace{1.5mm}\noindent
\textbf{Acknowledgements.}\medskip

D.M. and K.M. thank Robert Seiringer for helpful discussions. Financial support from the Agence Nationale de la Recherche (ANR) through the projects ANR-17-CE40-0016, ANR-17-CE40-0007-01, ANR-17-EURE-0002 (J.L.) and from the European Union’s Horizon 2020 research and innovation programme under the Maria Sk\l odowska-Curie grant agreement No. 665386 (K.M.) is gratefully acknowledged.
\noindent 


\end{spacing}

\end{document}